\documentclass[conference]{ieeeconf}
\IEEEoverridecommandlockouts
\usepackage{cite}
\usepackage{amsmath,amssymb,amsfonts,amsthm}
\usepackage{graphicx}
\usepackage{textcomp}
\usepackage{xcolor}
\usepackage[T1]{fontenc}
\let\labelindent\relax
\usepackage{enumitem}

\newtheorem{definition}{Definition}
\newtheorem{theorem}{Theorem}
\newtheorem{corollary}[theorem]{Corollary}
\newtheorem{lemma}[theorem]{Lemma}

\newcommand{\R}{\mathbb{R}}
\newcommand{\C}{\mathbf{C}}
\newcommand{\pq}{\phi_1}
\newcommand{\pw}{\phi_2}
\newcommand{\oxone}{\overline{X}_1}
\newcommand{\oxtwo}{\overline{X}_2}
\newcommand{\OG}{\overline{G}}
\newcommand{\xai}[1]{\overline{X}_{A, #1}}

\newcommand{\mysqrt}[1]{\left( #1 \right)^\frac{1}{2}}
\newcommand{\lcurl}[1]{\left\{ #1 \right\}}
\newcommand{\suchthat}{\,\ifnum\currentgrouptype=16 \middle\fi|\,}

\definecolor{myorange}{HTML}{FF9300}
\def\orangelozenge{\mathbin{\color{myorange}\blacklozenge}}

\newcommand{\UncertaintyGameNominal}{figs/subfigs/game_nominal_v2.pdf}
\newcommand{\UncertaintyGameAdversary}{figs/subfigs/game_adversary_uncertain.pdf}
\newcommand{\UncertaintyGameAlly}{figs/subfigs/game_ally_uncertain.pdf}
\newcommand{\UncertaintyPayoffNominal}{figs/subfigs/budget_sim_nominal.pdf}
\newcommand{\UncertaintyPayoffAdversary}{figs/subfigs/budget_sim_adversary.pdf}
\newcommand{\UncertaintyPayoffAlly}{figs/subfigs/budget_sim_ally.pdf}

\DeclareMathOperator*{\argmax}{arg\,max}

\def\BibTeX{{\rm B\kern-.05em{\sc i\kern-.025em b}\kern-.08em
    T\kern-.1667em\lower.7ex\hbox{E}\kern-.125emX}}

\title{\LARGE \bf Forming Alliances in the Fog of War: A General Lotto Perspective}

\author{Edik Hakobyan, Vade Shah, and Jason R. Marden
\thanks{This work is supported by}
\thanks{E. Hakobyan (\texttt{edikh7606@gmail.com}), V. Shah, and J. R. Marden are with the Department of Electrical and Computer Engineering at the University of California, Santa Barbara, CA.}
}

\begin{document}

\maketitle 

\begin{abstract}
Agents competing against a common opponent can often gain an edge by forming alliances, but deciding whether to do so is difficult when they lack precise knowledge about their opponent or ally. We study how uncertainty affects opportunities for alliance formation in the coalitional General Lotto game, in which two players compete aganist a shared adversary by allocating their resource budgets over separate sets of valued contests. When all agents know one another's budgets, it is known that a budget transfer from one player to the other can benefit both players. We ask whether such alliances survive when the players know the relevant budgets only up to a multiplicative factor. Under two models, one in which the players are uncertain about the adversary's budget and one in which they are uncertain about each other's budgets, we characterize exactly which games admit a transfer that benefits both players in every game consistent with what they conjecture to be possible. In both models, we show that such robust transfers exist in a set of games of positive measure for every level of uncertainty. These results suggest that alliances can remain worthwhile even when allies are largely in the dark.
\end{abstract}
 
\section{Introduction}

When agents compete against a shared opponent, they can often strengthen their positions by forming an alliance. Companies partner to gain an edge over competing products \cite{elmuti2001overview, culpan2002global, king2003complementary}, energy producers cooperate to thrive in markets \cite{van2015power, de2020cooperatives}, and countries unite against common threats \cite{gibler2008international}. At the heart of each of these partnerships often lies a sharing of resources, be it money, electricity, or military equipment, allowing allies to shore up vulnerabilities, build on one another's strengths, and even deter their opponents.

Choosing whether to form an alliance becomes considerably harder, however, when agents lack precise knowledge of the resources held by their opponents, or even by their partners. Such uncertainty arises in nearly every setting: alliance partners often hold private information about their own capabilities \cite{balakrishnan1993information, reuer2000asymmetric}; energy producers that rely on renewable generation cannot perfectly forecast how much surplus power they or their partners will have available \cite{bitar2012, baeyens2013coalitional}; and countries routinely act on imperfect intelligence regarding the military strength of both their rivals and their allies \cite{morrow1994alliances, smith1995alliance, fearon1995rationalist}. This raises a natural question: \textbf{How much uncertainty can an alliance tolerate before its benefits are no longer guaranteed?}

\begin{figure}
    \includegraphics[width=\linewidth]{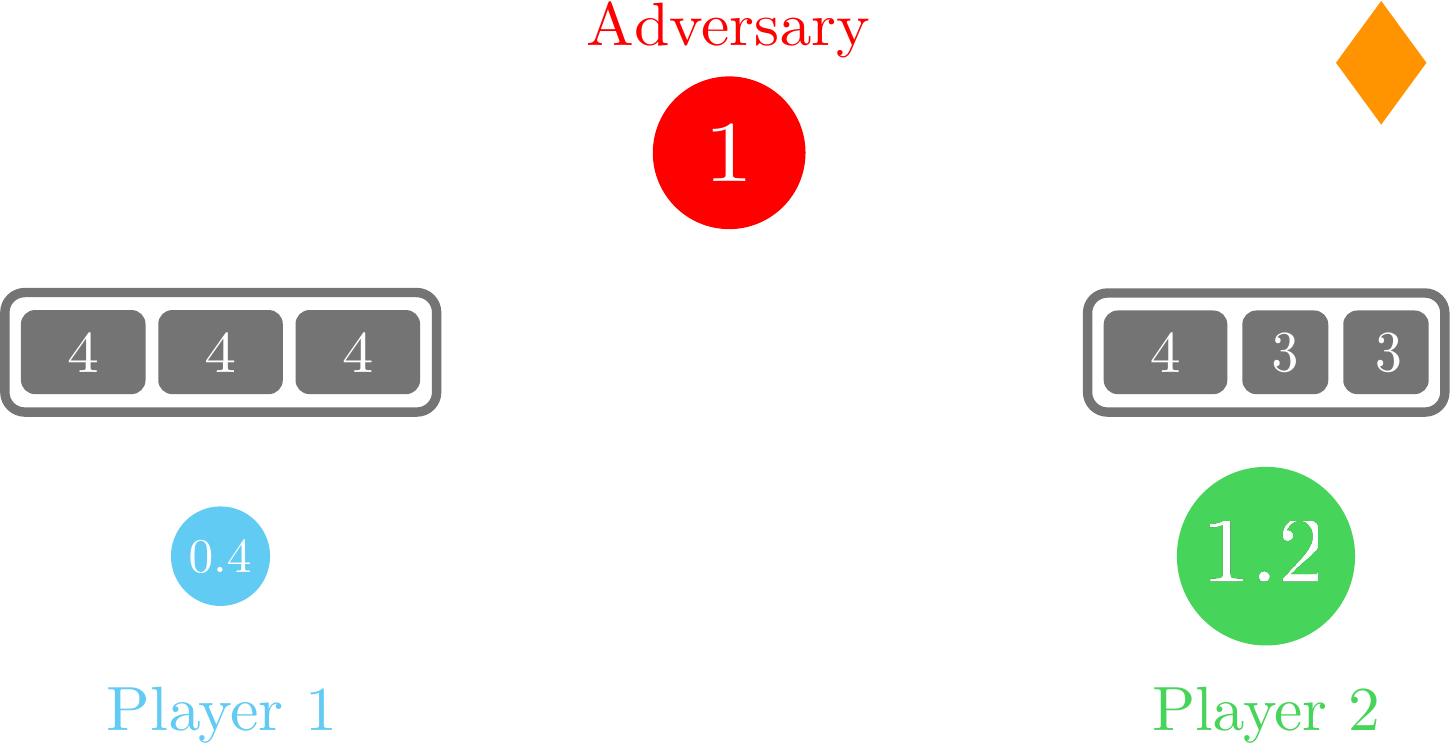}
    \caption{An illustration of a coalitional General Lotto game in which Players 1 and 2 face a shared adversary. The budgets of Player 1, Player 2, and the adversary are $0.4$, $1.2$, and $1$, respectively. Player 1 faces the adversary over the contests on the left, whose total value is $12$, while Player 2 faces the adversary over the contests on the right, whose total value is $10$. We use this game ($\orangelozenge$) as a running example throughout the paper.}
    \centering
    \label{fig:coalitional_basic}
\end{figure}

Game-theoretic tools have been widely used to assess the value of alliances across many adversarial domains \cite{shenoy1979coalition, smith1995alliance, sandler1999alliance, ray2015coalition}. Here, we concentrate on settings of \emph{competitive resource allocation} where agents strategically spread their resources to compete for prizes. A prominent model of competitive resource allocation is the \emph{General Lotto} game \cite{kovenock2021generalizations}, in which wo agents divide limited resources among a collection of valued contests. Whichever agent commits more resources to a contest wins its value, and each agent aims to maximize the total value they win while meeting their budget constraint in expectation.

To study alliance formation, we analyze the \emph{coalitional} General Lotto game, a well-established extension of the standard General Lotto game \cite{kovenock2012coalitional, gupta2014three, heyman2018colonel, chandan2022art, shah2024battlefield, shah2024inefficient, shah2026alliance}. In this game (Figure \ref{fig:coalitional_basic}), two players each face a common adversary in their own General Lotto game. Prior studies of this setting have uncovered a counterintuitive avenue for cooperation: in some games, when one player gives, or \emph{transfers}, part of their budget to the other, each player ends up winning more contest value in their own competition than they would have without the transfer \cite{kovenock2012coalitional}. In other words, even though the donor is left with fewer resources and the recipient with more, \emph{both} come out ahead. This occurs because the transfer prompts the adversary to reallocate their budget, allowing the players to steer the adversary's response in their favor.

Whether such transfers remain worthwhile when the players lack full knowledge of the game, however, is an open question. The mutually beneficial transfers identified in \cite{kovenock2012coalitional} rely on every agent knowing every budget exactly, and a transfer that benefits both players in one game may harm one of them if a budget has been misjudged. Moreover, while the role of information in one-versus-one General Lotto games has received considerable attention \cite{paarporn2021general, paarporn2024incomplete, brethouwer2024general}, far less is known about how uncertainty affects cooperation between allies. This work examines how robust these transfers are to uncertainty. Specifically, we consider two models in which the adversary observes the game exactly, but the players do not: \emph{adversary uncertainty}, in which the players know the adversary's budget only up to a multiplicative factor $1 \pm \delta$, and \emph{ally uncertainty}, in which each player knows the other's budget only up to the same factor. The parameter $\delta$ thus captures how much the allies are in the dark when deciding whether to cooperate. In this framework, we ask when budget transfers can still yield \emph{robust} mutually beneficial alliances, i.e., transfers that raise the equilibrium payoff of each player in every game consistent with what that player knows.

Our first main result shows that for any level of uncertainty about the adversary's budget, a robust mutually beneficial budget transfer exists in a nontrivial subset of games. More precisely, we characterize exactly when an alliance remains mutually beneficial despite uncertainty about the adversary. We then shift to the setting where the players are instead uncertain about one another's budgets. Our second main result shows that for any level of uncertainty about an ally's budget, a robust mutually beneficial budget transfer again exists in a nontrivial subset of games which we characterize exactly. Taken together, these results shed light on both the promise and the limits of forming alliances in uncertain competitive environments.

\section{Model}\label{sec:model}

In this section, we introduce the model of the coalitional General Lotto game in which two Players independently compete against a common Adversary across disjoint sets of contests. For ease of exposition, we begin with a discussion of the classical one-versus-one General Lotto game between a single Player and the Adversary.

\subsection{General Lotto Game}

In the \emph{General Lotto} game, two agents (the Player and the Adversary) compete across a set of $n$ contests. The Player and the Adversary are endowed with arbitrarily divisible resource budgets $X_P \in \R_{>0}$ and $X_A \in \R_{>0}$, respectively, which they allocate across the contests. The agent who allocates a greater level of resources towards the $k^\mathrm{th}$ contest wins its valuation $v^k \in \R_{>0}$ (ties are awarded to the Player). The budgets and the contest valuations are known to both agents, but neither knows how their opponent will allocate their budget across the contests. An allocation decision for the Player is a vector $D_P \in \R_{\geq 0}^n$, where $D_P^k$ is the amount of budget that they allocate to contest $k$; the decision for the Adversary $D_A$ is defined similarly. The Player's payoff for a given pair $(D_P, D_A)$ is of the form
\begin{equation*}
    U_P^{\rm GL}(D_P, D_A) = \sum_{k=1}^n v^k \cdot I \{D_{P}^k \geq D_{A}^k\},
\end{equation*}
where $I\{\cdot\}$ is the usual indicator function, and the Adversary's payoff is $U_{A}^{\rm GL}(D_{P}, D_{A}) = \sum_{k=1}^n v^k - U_{P}^{\rm GL}(D_P, D_A)$.

A \emph{strategy} for the Player is an $n$-variate distribution function $F_P$ satisfying
\begin{equation}
    \mathbb{E}_{D_P \sim F_P}\left[\sum_{k = 1}^n D_P^k \right] \leq X_P,
\end{equation}
meaning that the Player meets their budget constraint \emph{in expectation}. The Adversary's strategy $F_A$ is defined similarly.
The agents' Nash equilibrium strategies and payoffs in this setting are well-understood and have been characterized in \cite{kovenock2021generalizations}. For ease of exposition, we omit the equilibrium strategies and instead present only the equilibrium payoffs in the following Theorem from \cite{kovenock2021generalizations}:
\begin{theorem}[Kovenock and Roberson, 2021]
    For any General Lotto game with budgets $X_P$, $X_A$ and contest valuations $v^1, \dots, v^n$, a Nash equilibrium exists, and the unique Nash equilibrium payoffs for the Player and the Adversary are
    \begin{align}
        U_P^{\rm NE}(\phi, X_P, X_A) &\triangleq \begin{cases}
            \phi \left( \frac{X_P}{2 X_A} \right) & X_P \leq X_A \\
            \phi \left( 1 - \frac{X_A}{2 X_P} \right) & X_P > X_A \label{eq:player_payoff}
        \end{cases} \\
        \text{and } \; U_A^{\rm NE}(\phi, X_P, X_A) &\triangleq \phi - U_P^{\rm NE}(\phi, X_P, X_A), \label{eq:adversary_payoff}
    \end{align}
    respectively, where $\phi \triangleq \sum_{k = 1}^n v^k$.
\end{theorem}

\subsection{Coalitional General Lotto Game}

To study opportunities for cooperation, we analyze the \emph{coalitional General Lotto game}, a variant in which Players 1 and 2 compete against a common Adversary in separate standard General Lotto games. When convenient, we refer to one Player as $i \in \{1, 2\}$ and the other as $-i$. Our focus is on the formation of alliances between the Players, where an alliance is formed through the unilateral transfer of resources from one Player to the other. The interaction unfolds over multiple stages, as illustrated in Figure \ref{fig:stages} and described in detail below.

\begin{figure*}
    \includegraphics[width=\textwidth]{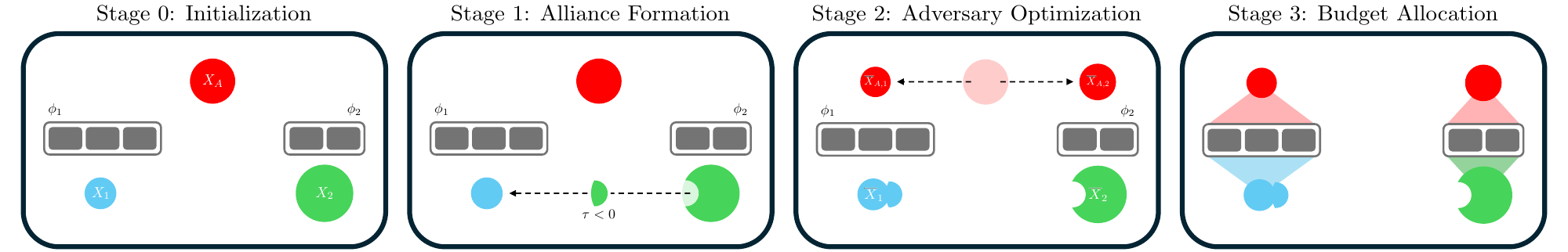}
    \caption{The stages of the standard coalitional General Lotto game. In Stage 0, the game is initialized. In Stage 1, the Players perform a transfer of resources. In Stage 2, the Adversary best responds to any transfers. In Stage 3, the agents allocate their resources toward the contests according to their Nash equilibrium strategies and receive their equilibrium payoffs.}
    \label{fig:stages}
\end{figure*}

\vspace{.2cm}
\noindent\textbf{Stage 0 (Initialization):} The game is initialized. Player 1, Player 2, and the Adversary have budgets $X_1$, $X_2$, and $X_A$, respectively. Player $i \in \{1, 2\}$ competes against the Adversary in General Lotto game $i$, where the $n_i$ contests have cumulative valuation $\phi_i$. A coalitional General Lotto game instance is thus fully parameterized by the tuple
\begin{equation}
    G = (\pq, \pw, X_1, X_2, X_A) \in \mathcal{G} = \R_{>0}^5
\end{equation}
which all agents observe before Stage 1.

\vspace{.2cm}
\noindent\textbf{Stage 1 (Alliance Formation):} The Players form an alliance through a budget transfer. Player 1 transfers 
\begin{equation}
    \tau \in (-X_2, X_1)
\end{equation}
units of their budget to Player 2; a negative value of $\tau$ indicates that the net amount is transferred from Player 2 to Player 1. The transfer yields new budgets
\begin{equation}
    \oxone = X_1 - \tau \; \text{ and } \; \oxtwo = X_2 + \tau
\end{equation}
for Players 1 and 2, respectively.

\vspace{.2cm}
\noindent\textbf{Stage 2 (Adversary Optimization):} The Adversary best responds to any transfers between the Players by optimally dividing their budget between the two General Lotto games. Mathematically, this entails solving the problem
\begin{equation}\label{eq:adversary_optimization}
    \argmax_{\underset{X_{A,1} + X_{A,2} \leq X_A}{X_{A,1}, X_{A,2} \geq 0}} U_A^{\rm NE}(\phi_1, \oxone, X_{A,1}) + U_A^{\rm NE}(\phi_2, \oxtwo, X_{A,2}),
\end{equation}
where $X_{A, i}$ is the amount of budget they send to Lotto game $i$.
We denote the solution pair to problem \eqref{eq:adversary_optimization} as $(\xai{1}(G, \tau), \xai{2}(G, \tau))$, which has been derived in \cite{kovenock2012coalitional} and summarized in Appendix \ref{app:cases}\footnote{The solution pair is unique whenever $\oxone / \phi_1 \neq \oxtwo / \phi_2$, i.e., for every transfer except $\tau = (\phi_2 X_1 - \phi_1 X_2)/(\phi_1 + \phi_2)$. At this single transfer, the Adversary may be indifferent among multiple allocations that yield different payoffs for the Players; since this constitutes a measure-zero set of transfers, we take $(\xai{1}(G, \tau), \xai{2}(G, \tau))$ to be a fixed value from the set of solutions for all $(G, \tau)$.}.

\vspace{.2cm}
\noindent\textbf{Stage 3 (Budget Allocation):} In the third and final stage, the Players and the Adversary allocate their budgets towards the contests according to their equilibrium General Lotto strategies and receive their corresponding equilibrium payoffs \eqref{eq:player_payoff}, \eqref{eq:adversary_payoff}. To highlight the dependence on the transfer amount, we rewrite Player 1 and 2's utility functions as
\begin{align}
    U_1(G, \tau) &\triangleq U_P^{\rm NE}(\phi_1, X_1 - \tau, \overline{X}_{A, 1}(G, \tau)) \; \text{ and } \\
    U_2(G, \tau) &\triangleq U_P^{\rm NE}(\phi_2, X_2 + \tau, \overline{X}_{A, 2}(G, \tau)),
\end{align}
respectively.

In this work, we are interested in budget transfers that strictly increase both Players' payoffs:

\begin{definition}
    For a given coalitional General Lotto game $G \in \mathcal{G}$, a transfer $\tau \in (-X_2, X_1)$ is \emph{mutually beneficial} if
    \begin{equation}
        U_1(G, \tau) > U_1(G, 0) \; \text{ and } \; U_2(G, \tau) > U_2(G, 0).
    \end{equation}
\end{definition}

Previous work has demonstrated the counterintuitive result that mutually beneficial transfers exist in a nontrivial subset of games \cite{kovenock2012coalitional}. For ease of presentation, we focus throughout on games satisfying
\begin{equation}\label{eq:ratio-assumption}
    \frac{X_2}{\phi_2} \geq \frac{X_1}{\phi_1},
\end{equation}
i.e., the ratio of Player 2's budget to contest valuation is at least that of Player 1; all of our results extend to games with $X_2 / \phi_2 < X_1 / \phi_1$ by swapping indices. For such games, the following result combines Propositions 1 and 2 of \cite{kovenock2012coalitional}; its two cases correspond to Cases 2 and 3 in Appendix~\ref{app:cases}.

\begin{lemma}[Kovenock and Roberson, 2012]\label{lem:nominal}
    Consider any game $G \in \mathcal{G}$ that satisfies $\frac{X_2}{\phi_2} \geq \frac{X_1}{\phi_1}$.
    \begin{enumerate}[label=(\roman*)]
        \item If $\mysqrt{\frac{\phi_1 X_1 X_2}{\phi_2}} < X_A \leq X_2 + \mysqrt{\frac{\phi_1 X_1 X_2}{\phi_2}}$, then $G$ has a mutually beneficial transfer if and only if
        \begin{equation}\label{eq:nominal-case-2}
            \frac{X_1 + X_2}{X_A} > 2 \mysqrt{\frac{\phi_2 X_1}{\phi_1 X_2}}.
        \end{equation}
        \item If $X_A > X_2 + \mysqrt{\frac{\phi_1 X_1 X_2}{\phi_2}}$, then $G$ has a mutually beneficial transfer if and only if
        \begin{equation}\label{eq:nominal-case-3}
            1 - \frac{X_1}{X_2} > 2 \mysqrt{\frac{\phi_2 X_1}{\phi_1 X_2}}.
        \end{equation}
    \end{enumerate}
    Further, no transfer $\tau > 0$ is mutually beneficial, i.e., any mutually beneficial transfer flows from Player 2 to Player 1.
\end{lemma}

\section{Results}

The focus of this work is on understanding if Players can still form mutually beneficial alliances through budget transfers when they are \emph{uncertain} about the parameters of the game. In this section, we consider two uncertainty models and analyze whether mutually beneficial transfers persist under these models.

Throughout, we consider games satisfying \eqref{eq:ratio-assumption}, for which any mutually beneficial transfer flows from Player 2 to Player 1 (i.e., $\tau < 0$) when there is no uncertainty (Lemma~\ref{lem:nominal}). Further, for notational convenience, given a game $G$ and parameter $P$, we write $G \mid_{P \mapsto P'}$ for the game that is identical to $G$ except that $P$ is replaced by $P'$.

\begin{figure*}
    \centering
    \begin{tabular}{@{}c@{\hspace{0.04\linewidth}}c@{\hspace{0.04\linewidth}}c@{}}
        \includegraphics[width=0.30\linewidth]{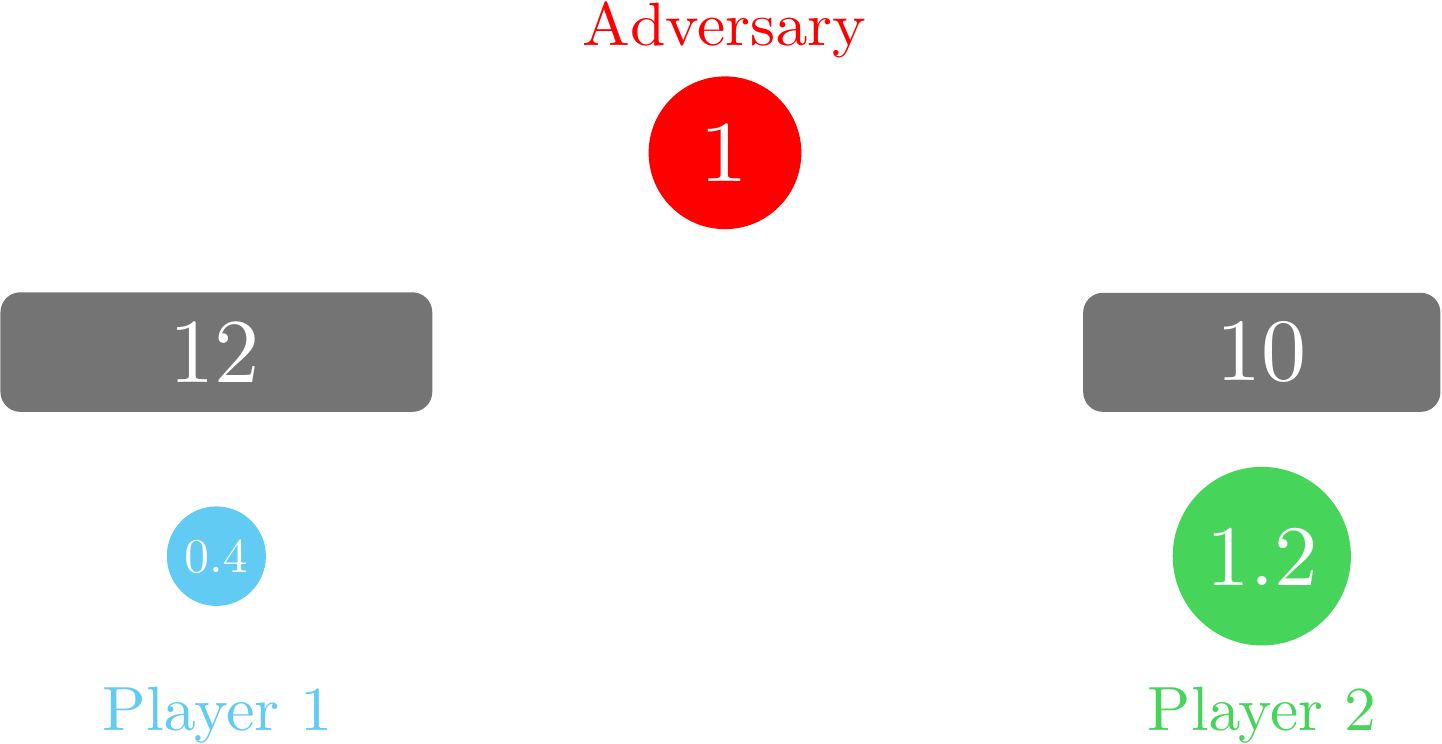} &
        \includegraphics[width=0.30\linewidth]{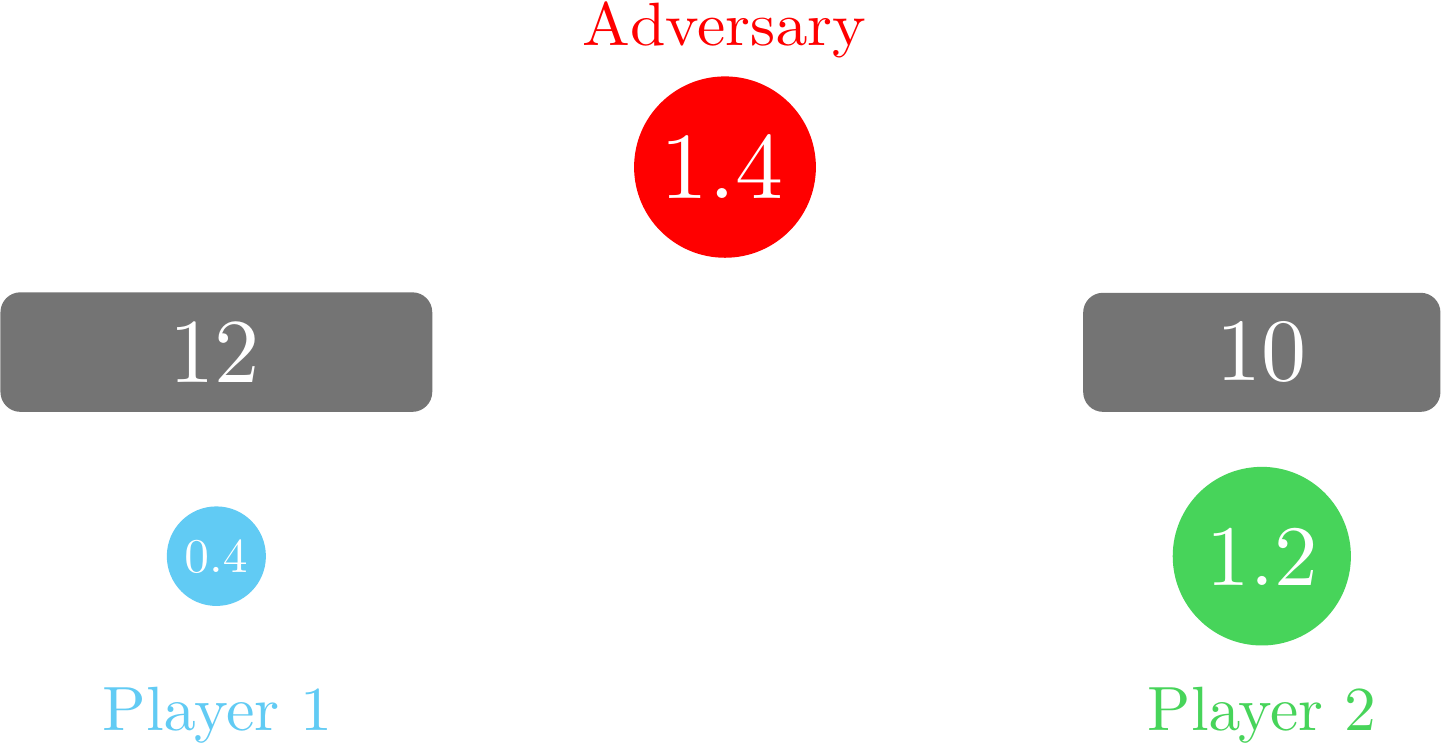} &
        \includegraphics[width=0.30\linewidth]{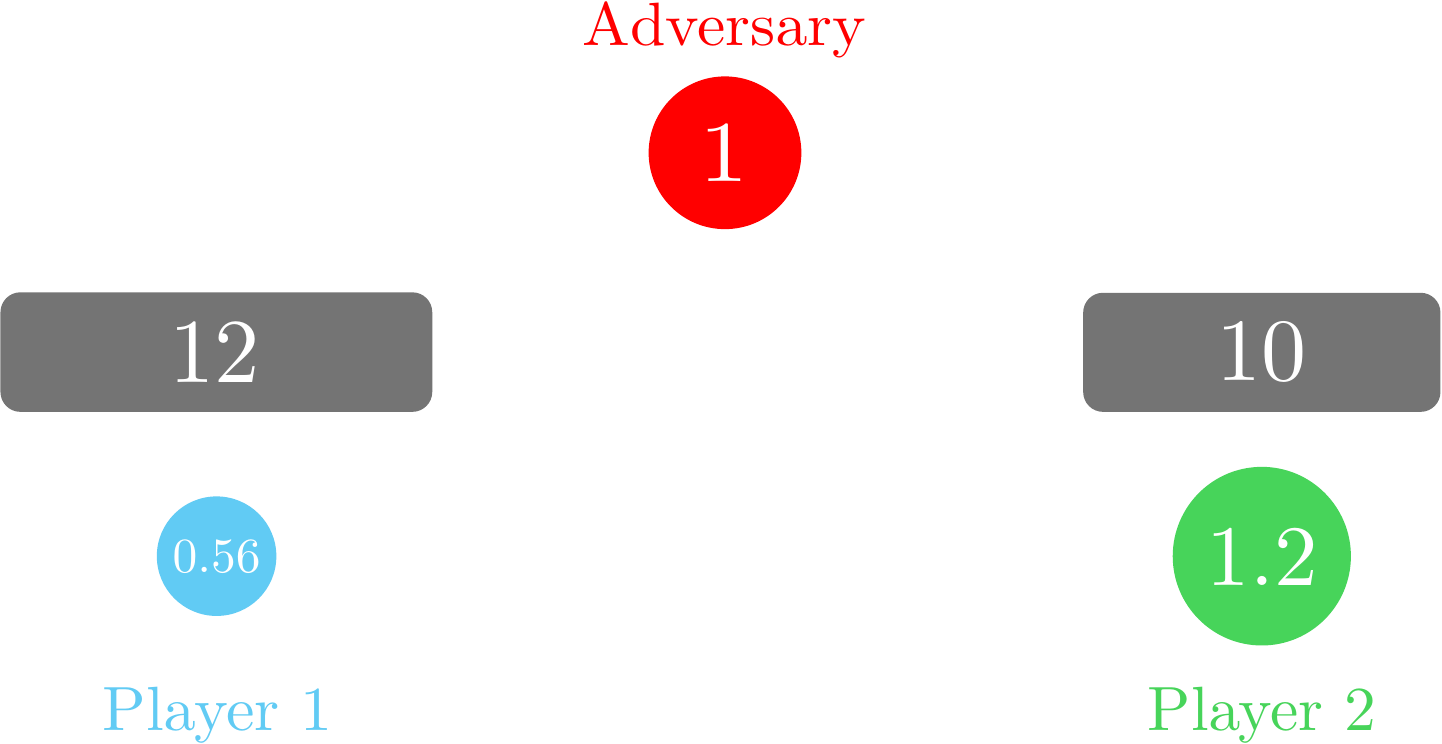} \\[0.4em]
        \includegraphics[width=0.30\linewidth]{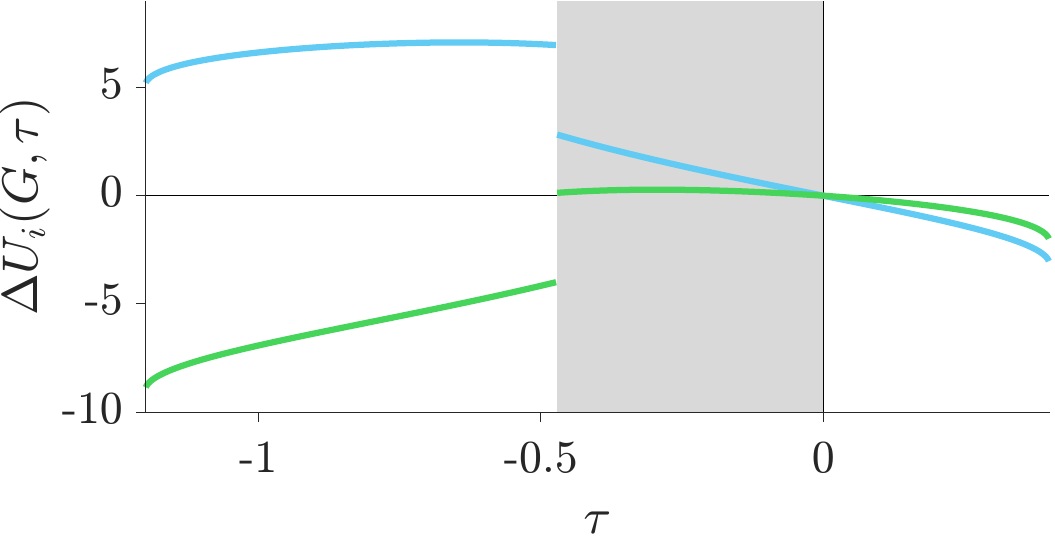} &
        \includegraphics[width=0.30\linewidth]{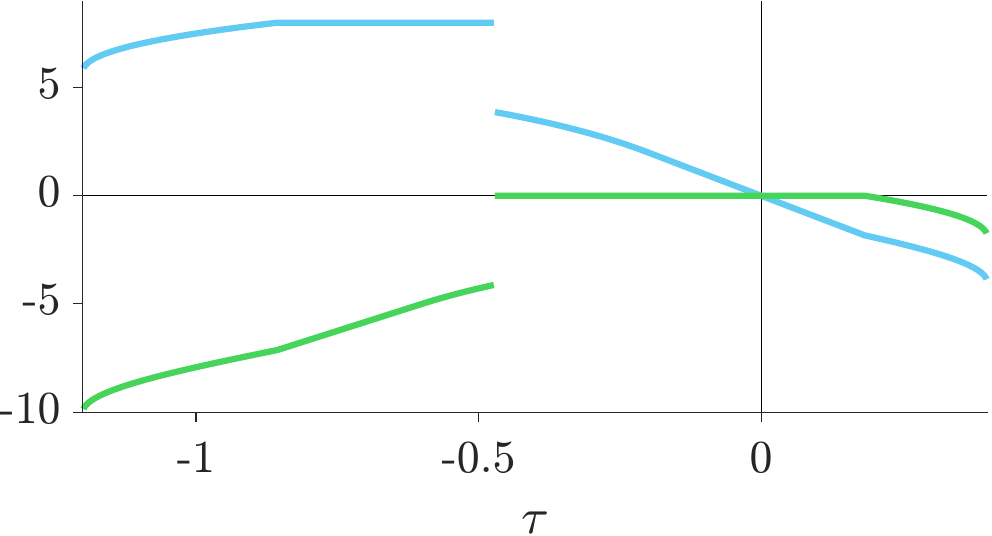} &
        \includegraphics[width=0.30\linewidth]{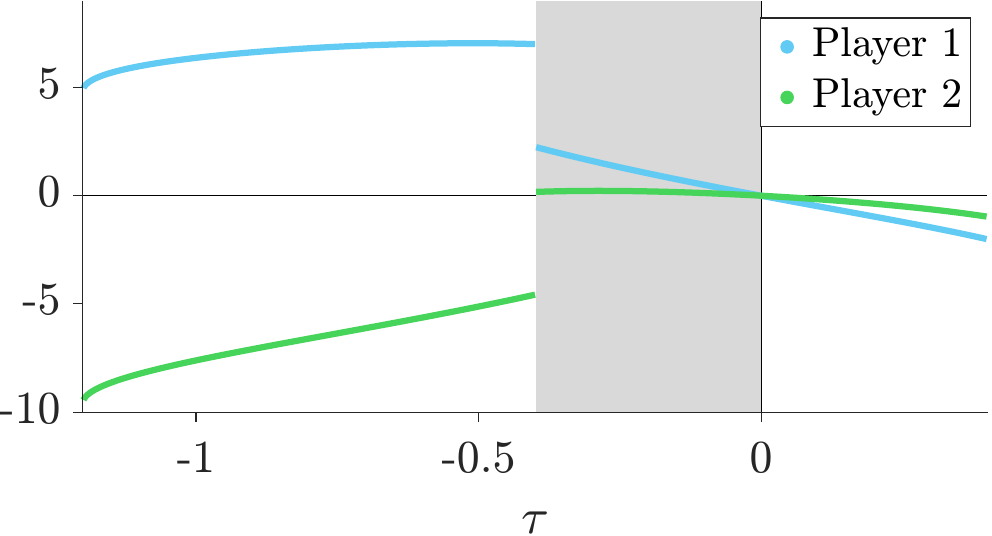}
    \end{tabular}
    \caption{Player payoffs as a function of the transfer $\tau$ in a nominal game $G = (\phi_1, \phi_2, X_1, X_2, X_A) = (12, 10, 0.4, 1.2, 1)$ (left), a perturbed game that differs from $G$ only in the Adversary's budget, $X_A' = X_A(1-\delta) = 0.6$ (middle), and a perturbed game that differs from $G$ only in Player~1's budget, $X_1' = X_1(1+\delta) = 0.56$ (right), with $\delta = 0.4$. Top row: game depictions. Bottom row: each Player's change in payoff $\Delta U_i(G, \tau) \triangleq U_i(G, \tau) - U_i(G, 0)$ in the depicted game; shaded regions mark mutually beneficial transfers. The middle game lies in ${\cal G}^{X_A}(G, 0.4)$ and admits no mutually beneficial transfer, so $G$ has no $0.4$-robust mutually beneficial transfer under Adversary uncertainty. The right panel shows that mutual benefit can persist under ally uncertainty; in fact $G$ is $0.4$-robust to ally uncertainty, though verifying this requires the full set ${\cal G}^{X_{-i}}(G, 0.4)$, not just the single realization depicted.}
    \label{fig:uncertainty-example}
\end{figure*}

\subsection{Adversary Uncertainty}

First, we consider the setting where both Players are uncertain about the Adversary's budget. In contrast to the model in Section~\ref{sec:model}, in which all agents observe $G$ before Stage 1, the Adversary continues to observe $G$ exactly, but the Players now observe only some of its parameters. In particular, the Players know only that the game belongs to the set ${\cal G}^{X_A}(G, \delta)$ given by
\begin{equation}
    {\cal G}^{X_A}(G, \delta) \triangleq \left\{G |_{X_A \mapsto X_A'} : 1 - \delta \leq \frac{X_A'}{X_A} \leq 1 + \delta \right\}.
\end{equation}
Here, $\delta \in [0, 1)$ parameterizes the Players' uncertainty about the Adversary's budget; when $\delta = 0$, this recovers the nominal case with no uncertainty (${\cal G}^{X_A}(G, 0) = \{G\}$). Our goal is to understand whether there exist mutually beneficial transfers that are \emph{robust} to this uncertainty:

\begin{definition}
    Given a game $G \in \mathcal{G}$, we will say that a mutually beneficial transfer $\tau$ is $\delta$-robust to Adversary uncertainty if for each Player $i \in \{1,2\}$
    \begin{equation}
        U_i(G_i, \tau) > U_i(G_i, 0)
    \end{equation}
    for all games $G_i \in {\cal G}^{X_A}(G, \delta)$.
\end{definition}
\noindent Figure~\ref{fig:uncertainty-example} shows an example of a game that has a mutually beneficial transfer, but does not have a transfer that is $0.4$-robust to Adversary uncertainty.

Our first main result is a complete characterization of the set of games that admit robust mutually beneficial transfers.

\begin{theorem}\label{thm:adversary-uncertainty}
    Consider any game $G \in \mathcal{G}$ that has a mutually beneficial transfer and satisfies $\frac{X_2}{\phi_2} \geq \frac{X_1}{\phi_1}$. Then, $G$ also has a transfer that is $\delta$-robust to Adversary uncertainty if and only if
    \begin{equation}\label{eq:adversary-uncertainty-cond-1}
        \delta < \begin{cases}
            \alpha_1(G) & X_2 - X_1> 2 \mysqrt{\frac{\phi_2 X_1 X_2}{\phi_1}}, \\
            \min \{\alpha_1(G), \alpha_2(G)\} & \text{otherwise,}
        \end{cases}
    \end{equation}
    where
    \begin{align}
        \alpha_1(G) &\triangleq 1 - \frac{1}{X_A}\mysqrt{\frac{\phi_1 X_1 X_2}{\phi_2}}, \label{eq:adversary-uncertainty-cond-2} \\
        \alpha_2(G) &\triangleq \frac{X_1 + X_2}{2 X_A} \mysqrt{\frac{\phi_1 X_2}{\phi_2 X_1}} - 1. \label{eq:adversary-uncertainty-cond-3}
    \end{align}
\end{theorem}

The significance of Theorem \ref{thm:adversary-uncertainty} (proven in Appendix \ref{app:adversary}) is not the condition \eqref{eq:adversary-uncertainty-cond-1}--\eqref{eq:adversary-uncertainty-cond-3} itself, but rather that the set of games admitting a $\delta$-robust mutually beneficial transfer has positive measure for every $\delta \in [0, 1)$; this set is illustrated for several values of $\delta$ in Figure~\ref{fig:robust-adversary-region}.

\begin{corollary}\label{cor:adversary-uncertainty}
    For any $\delta \in [0, 1)$, the set of games $G \in \mathcal{G}$ that admit a transfer that is $\delta$-robust to Adversary uncertainty has positive measure.
\end{corollary}

\begin{proof}
    We first construct one such game. Fix $\delta \in [0, 1)$, the valuations, and $X_A$. Take $X_2 = \kappa X_1$, where $\kappa$ satisfies $(\kappa - 1)/\kappa^{1/2} > 2(\phi_2/\phi_1)^{1/2}$. Since $(\kappa - 1)/\kappa^{1/2} < \kappa^{1/2}$, this gives $\kappa > \phi_2/\phi_1 \iff X_2/\phi_2 > X_1/\phi_1$. Next, take $X_1$ small enough that $X_1(\kappa + (\kappa\phi_1/\phi_2)^{1/2}) < X_A$, which places $G$ in case (ii) of Lemma~\ref{lem:nominal}. Since
    \begin{equation*}
        X_2 - X_1 = (\kappa - 1) X_1 > 2 \mysqrt{\frac{\kappa \phi_2}{\phi_1}} X_1 = 2 \mysqrt{\frac{\phi_2 X_1 X_2}{\phi_1}},
    \end{equation*}
    which is equivalent to \eqref{eq:nominal-case-3}, Lemma~\ref{lem:nominal} implies that $G$ has a mutually beneficial transfer, and \eqref{eq:adversary-uncertainty-cond-1} reduces to $\delta < \alpha_1(G)$. Moreover,
    \begin{equation*}
        \alpha_1(G) = 1 - \frac{X_1}{X_A}\mysqrt{\frac{\kappa \phi_1}{\phi_2}} \;\to\; 1 \quad \text{as } X_1 \to 0,
    \end{equation*}
    so decreasing $X_1$ further if needed gives $\delta < \alpha_1(G)$. Finally, each inequality used above is strict and continuous in the game parameters, so all of them continue to hold on an open neighborhood of $G$ in $\mathcal{G}$. By Lemma~\ref{lem:nominal} and Theorem~\ref{thm:adversary-uncertainty}, every game in this neighborhood admits a transfer that is $\delta$-robust to Adversary uncertainty, so this set has positive measure.
\end{proof}

This result establishes that there is no prohibitively large Adversary uncertainty beyond which opportunities for alliance formation vanish completely. As such, this result suggests that Players should still consider alliance formation as a strategic mechanism, even when they have almost no information about their opponent's resource levels.

\begin{figure}
    \centering
    \includegraphics[width=\linewidth]{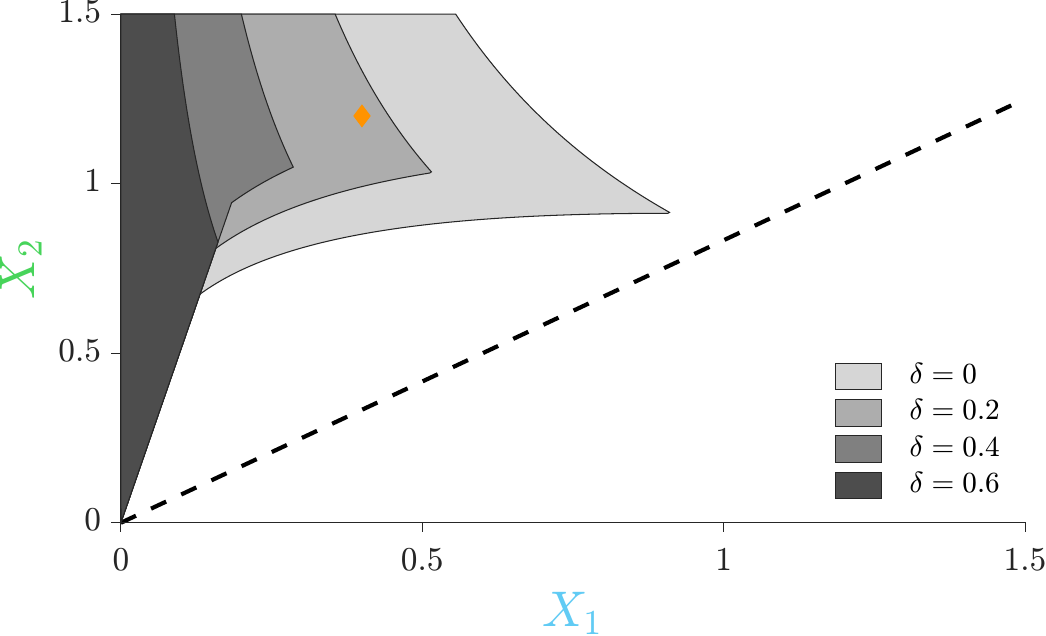}
    \caption{Games in the $X_1$--$X_2$ plane (with $\phi_1 = 12$, $\phi_2 = 10$, $X_A = 1$) that admit a $\delta$-robust mutually beneficial transfer under Adversary uncertainty, for $\delta \in \{0, 0.2, 0.4, 0.6\}$. Lighter regions correspond to smaller $\delta$ (which contain darker regions of larger $\delta$). The dashed line is $X_2/X_1 = \phi_2/\phi_1$. The diamond ($\orangelozenge$) marks the game in Figure~\ref{fig:coalitional_basic}; consistent with Figure~\ref{fig:uncertainty-example}, it admits a transfer that is $0.2$-robust, but none that is $0.4$-robust, to Adversary uncertainty.}
    \label{fig:robust-adversary-region}
\end{figure}

\subsection{Ally Uncertainty}

Next, we consider the setting where the Players may be uncertain of one another's budgets. As before, the Adversary knows the true game parameters $G$ exactly, but the Players do not: each Player $i \in \{1, 2\}$ knows only that the game belongs to the set
\begin{equation}\label{eq:ally-uncertainty-set}
    {\cal G}^{X_{-i}}(G, \delta)
    \triangleq
    \left\{
        G|_{X_{-i} \mapsto X_{-i}'}
        \;\middle|\;
        1 - \delta \leq \frac{X_{-i}'}{X_{-i}} \leq 1 + \delta
    \right\}.
\end{equation}
Here, $\delta \in [0, 1)$ captures how uncertain each Player is about the other's budget; for simplicity, we assume a common $\delta$, but neither Player knows the other's uncertainty set (i.e., Player $i$ does not know that Player $-i$ also has uncertainty of the form \eqref{eq:ally-uncertainty-set}).

It turns out that the direction of the transfer $\tau$ itself allows the Players to further reason about one another's budgets. Recall from Lemma~\ref{lem:nominal} that mutually beneficial transfers always flow from the Player with the larger budget-to-valuation ratio to the Player with the smaller one, and the giving Player only benefits if their ratio remains the larger one after the transfer \cite{kovenock2012coalitional}. For this uncertainty model, this result implies that Player 2 only offers resources to Player 1 if they are certain that their post-transfer ratio $\overline{X}_2/\phi_2$ is greater than $\overline{X}_1'/\phi_1$ for any $X_1' \leq (1 + \delta) X_1$. In turn, if Player 1 is offered resources by Player 2, Player 1 can infer that $\overline{X}_2/\phi_2$ is greater than their own post-transfer ratio $\overline{X}_1/\phi_1$. Hence, we define the updated uncertainty sets
\begin{align}
    {\cal G}^{X_{-i}}(G, \delta, \tau)
    &\triangleq
    \left\{ G' \in {\cal G}^{X_{-i}}(G, \delta) : \overline{X}_{-i}' \geq \ell_i(G, \tau) \right\}, \nonumber \\
    \ell_i(G, \tau)
    &\triangleq
    \begin{cases}
        (1 - \delta) X_{-i} & (-1)^i \tau \leq 0, \\
        \frac{\phi_{-i}}{\phi_i} \overline{X}_i & \text{otherwise},
    \end{cases}
\end{align}
With this augmentation, we can now define robust mutually beneficial transfers as before:

\begin{definition}
    Given a game $G \in \mathcal{G}$, we will say that a mutually beneficial transfer $\tau$ is $\delta$-robust to ally uncertainty if for each Player $i \in \{1,2\}$
    \begin{equation}
        U_i(G_i, \tau) > U_i(G_i, 0)
    \end{equation}
    for all games $G_i \in {\cal G}^{X_{-i}}(G, \delta, \tau)$.
\end{definition}
\noindent Figure~\ref{fig:uncertainty-example} illustrates a game with a transfer that is $0.4$-robust to ally uncertainty.

Our next results are of the same flavor as Theorem~\ref{thm:adversary-uncertainty} and Corollary~\ref{cor:adversary-uncertainty} for the ally uncertainty model:

\begin{theorem}\label{thm:ally-uncertainty}
    Consider any game $G \in \mathcal{G}$ that has a mutually beneficial transfer and satisfies $\frac{X_2}{\phi_2} \geq \frac{X_1}{\phi_1}$. Then, $G$ also has a transfer that is $\delta$-robust to ally uncertainty if and only if
    \begin{equation}\label{eq:ally-uncertainty-cond-1}
        \delta < \begin{cases}
            \min\{\beta_1(G), \beta_2(G), \beta_4(G)\} & X_A \leq \beta_5(G) X_2, \\
            \min\{\beta_1(G), \beta_2(G), \beta_6(G)\} & \text{otherwise},
        \end{cases}
    \end{equation}
    where
    \begin{align*}
        \beta_1(G) &\triangleq \frac{\phi_1 X_2}{\phi_2 X_1} - 1, \quad
        \beta_2(G) \triangleq \frac{\phi_2 X_A^2}{\phi_1 X_1 X_2} - 1, \\
        \beta_3(G) &\triangleq \mysqrt{\beta_2(G) + 1} - \mysqrt{\beta_2(G) + 1 - \frac{X_2}{X_1}}, \\
        \beta_4(G) &\triangleq \begin{cases}
            \bigl(\beta_3(G)\bigr)^2 - 1 & X_A \geq X_2 \mysqrt{\phi_1 / \phi_2}, \\
            +\infty & \text{otherwise},
        \end{cases} \\
        \beta_5(G) &\triangleq \mysqrt{1 + \phi_1 / \phi_2}, \quad
        \beta_6(G) \triangleq \frac{(\beta_5(G) - 1) X_2}{(\beta_5(G) + 1) X_1} - 1.
    \end{align*}
\end{theorem}

As with Theorem~\ref{thm:adversary-uncertainty}, the significance of Theorem~\ref{thm:ally-uncertainty} (proven in Appendix \ref{app:ally}) lies not in the specific form of \eqref{eq:ally-uncertainty-cond-1}, but in the fact that the set of games admitting a transfer that is $\delta$-robust to ally uncertainty has positive measure for every $\delta \in [0, 1)$; this set is illustrated for several values of $\delta$ in Figure~\ref{fig:robust-ally-region}.

\begin{corollary}\label{cor:ally-uncertainty}
    For any $\delta \in [0, 1)$, the set of games $G \in \mathcal{G}$ that admit a transfer that is $\delta$-robust to ally uncertainty has positive measure.
\end{corollary}

\begin{proof}
    We first construct one such game. Fix $\delta \in [0, 1)$, the valuations, and $X_1 > 0$. Note that $\beta_5(G) = (1 + \phi_1/\phi_2)^{1/2} > 1$ depends only on the valuations, and define
    \begin{equation*}
        \kappa_0 \triangleq \frac{\beta_5(G) + 1}{\beta_5(G) - 1} = \bigl(\beta_5(G) + 1\bigr)^2 \frac{\phi_2}{\phi_1}.
    \end{equation*}
     Take $X_2 = \kappa X_1$, where $\kappa > (1 + \delta) \kappa_0$. Since $\kappa > \kappa_0 > \phi_2/\phi_1$, we have $X_2/\phi_2 > X_1/\phi_1$. Next, take $X_A$ large enough that $X_A > X_2 + (\phi_1 X_1 X_2/\phi_2)^{1/2}$ and $X_A > \beta_5(G) X_2$. The first inequality places $G$ in case (ii) of Lemma~\ref{lem:nominal}. Moreover, $\kappa_0^{1/2} - \kappa_0^{-1/2} = \bigl((\beta_5(G) + 1) - (\beta_5(G) - 1)\bigr)(\phi_2/\phi_1)^{1/2} = 2(\phi_2/\phi_1)^{1/2}$, and $\kappa \mapsto \kappa^{1/2} - \kappa^{-1/2}$ is increasing, so
    \begin{equation*}
        X_2 - X_1 = (\kappa - 1) X_1 > 2 \mysqrt{\frac{\kappa \phi_2}{\phi_1}} X_1 = 2 \mysqrt{\frac{\phi_2 X_1 X_2}{\phi_1}}.
    \end{equation*}
    This is equivalent to \eqref{eq:nominal-case-3}, so Lemma~\ref{lem:nominal} implies that $G$ has a mutually beneficial transfer. The second inequality on $X_A$ means that \eqref{eq:ally-uncertainty-cond-1} reduces to $\delta < \min\{\beta_1(G), \beta_2(G), \beta_6(G)\}$. Of these, only $\beta_6(G)$ is relevant, since
    \begin{align*}
        \beta_6(G) + 1 &= \frac{\kappa}{\kappa_0} = \frac{\kappa \phi_1}{\phi_2 (\beta_5(G) + 1)^2} < \frac{\kappa \phi_1}{\phi_2} = \beta_1(G) + 1, \\
        \beta_2(G) + 1 &= \frac{\phi_2 X_A^2}{\phi_1 X_1 X_2} > \frac{\phi_2 \beta_5(G)^2 X_2}{\phi_1 X_1} \\
        &= \Bigl(1 + \frac{\phi_2}{\phi_1}\Bigr) \kappa > \frac{\kappa}{\kappa_0} = \beta_6(G) + 1.
    \end{align*}
    Moreover, $\beta_6(G) = \kappa/\kappa_0 - 1 > (1 + \delta) - 1 = \delta$. Finally, each inequality used above is strict and continuous in the game parameters, so all of them continue to hold on an open neighborhood of $G$ in $\mathcal{G}$. By Lemma~\ref{lem:nominal} and Theorem~\ref{thm:ally-uncertainty}, every game in this neighborhood admits a transfer that is $\delta$-robust to ally uncertainty, so this set has positive measure.
\end{proof}

Similar to Corollary \ref{cor:adversary-uncertainty}, this result establishes that there is no prohibitively large Adversary uncertainty beyond which opportunities for alliance formation vanish completely. As such, this result suggests that Players should still consider alliance formation as a strategic mechanism, even when they have almost no information about one another's resource levels.

\begin{figure}
    \centering
    \includegraphics[width=\linewidth]{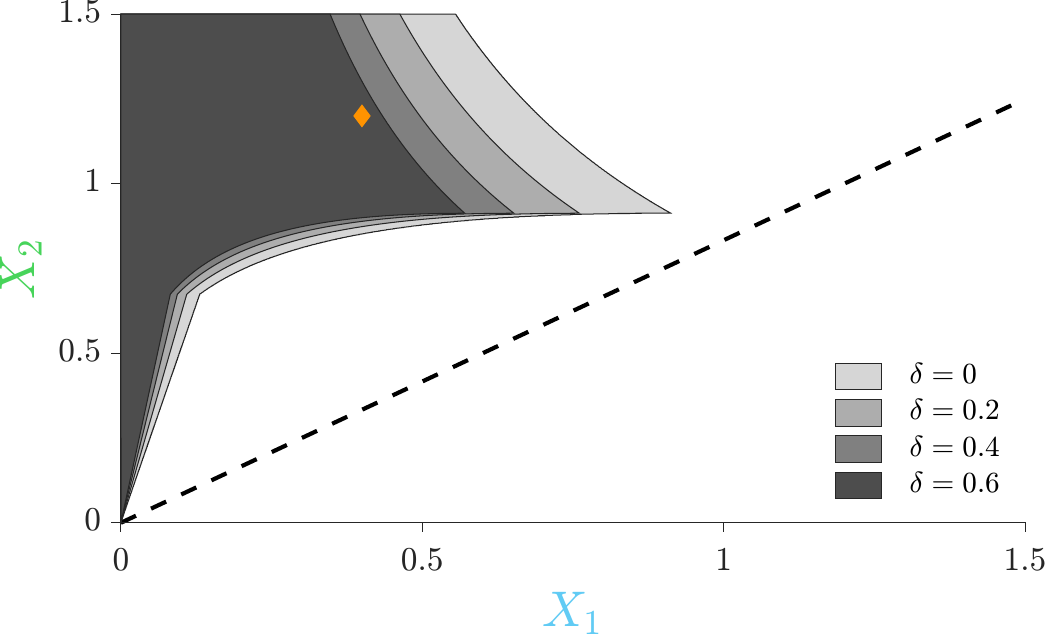}
    \caption{Games in the $X_1$--$X_2$ plane (with $\phi_1 = 12$, $\phi_2 = 10$, $X_A = 1$) that admit a $\delta$-robust mutually beneficial transfer under ally uncertainty, for $\delta \in \{0, 0.2, 0.4, 0.6\}$. Lighter regions correspond to smaller $\delta$ (which contain darker regions of larger $\delta$). The dashed line is $X_2/X_1 = \phi_2/\phi_1$. The diamond ($\orangelozenge$) marks the game in Figure~\ref{fig:coalitional_basic}, which admits a transfer that is $0.6$-robust to ally uncertainty.}
    \label{fig:robust-ally-region}
\end{figure}

\section{Conclusion}

In this work, we examine a multi-stage coalitional General Lotto game in which two players compete against a common adversary by allocating their limited budgets towards valued contests, but lack precise knowledge of the budgets held by their opponent or their ally. We first consider players who are uncertain about the adversary's budget, and we characterize exactly when a budget transfer remains mutually beneficial in every game consistent with the players' knowledge, showing that such robust alliances exist in a nontrivial subset of games for any level of uncertainty. Then, we consider players who are instead uncertain about one another's budgets, and we show that robust mutually beneficial transfers again exist in a nontrivial subset of games for any level of uncertainty. These results lend novel insight into the robustness of alliances to uncertainty and prompt further investigation into the effects of uncertainty in more general networks of competition, such as those involving many allies or multiple adversaries.

\appendix
    
\subsection{Adversary's Best Response}\label{app:cases}

The Adversary's best response to any transfer performed by the Players can be completely characterized by seven distinct Cases \cite{kovenock2012coalitional}. Four of these Cases address the setting where $\frac{\oxone}{\pq} \leq \frac{\oxtwo}{\pw}$, meaning that the ratio of Player 1's budget to contest valuation (after a transfer has been performed) is less than that of Player 2; the remaining three Cases address the setting where $\frac{\oxone}{\pq} > \frac{\oxtwo}{\pw}$ and can be obtained by a simple swapping of indices. Hence, without loss of generality, we describe below the Adversary's behavior for the four Cases that comprehensively describe all games
\begin{equation}
    \OG \in \mathcal{G}_{1 \leq 2} \triangleq \lcurl{G \in \mathcal{G} \suchthat \frac{\oxone}{\pq} \leq \frac{\oxtwo}{\pw}}
\end{equation}
as illustrated in Figure \ref{fig:cases}.

\begin{figure}
    \centering
    \includegraphics[width=\linewidth]{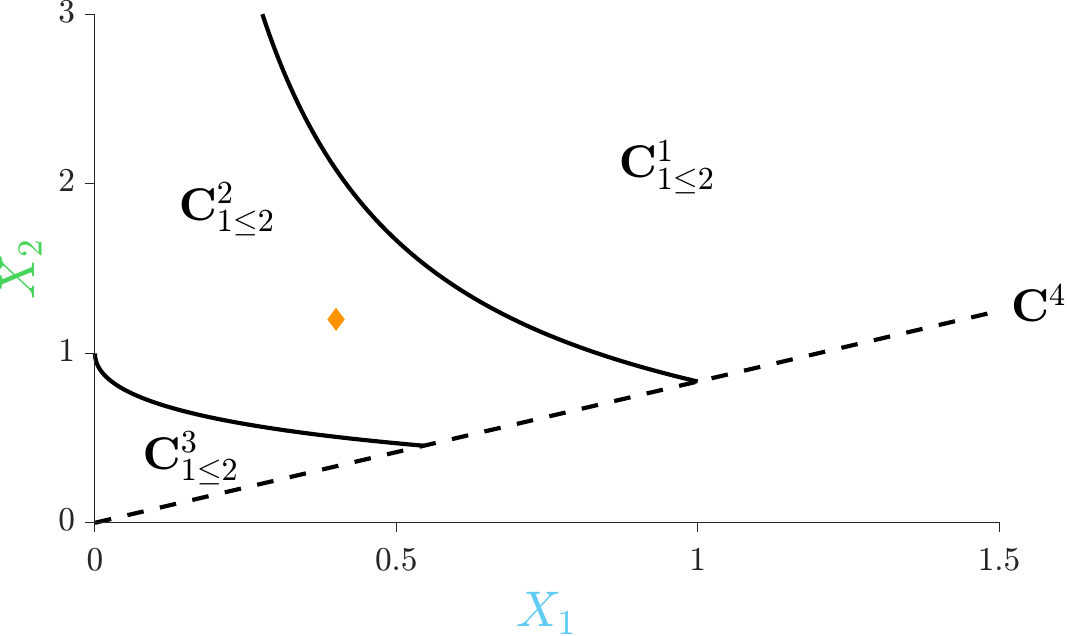}
    \caption{The Cases that delineate the form of the Adversary's best response, plotted in the $X_1$-$X_2$ space for fixed $\phi_1 = 12$, $\phi_2 = 10$, $X_A = 1$. The diamond ($\orangelozenge$) marks the game in Figure~\ref{fig:coalitional_basic}, which lies in Case 2.}
    \label{fig:cases}
\end{figure}

\vspace{.2cm}
\noindent\emph{Case 1:} If $\OG$ belongs to
    \begin{equation}\label{eq:case_1}
        \C_{1 \leq 2}^1 \triangleq \lcurl{\OG \in \mathcal{G}_{1 \leq 2} \suchthat
        \begin{gathered}
            \frac{\oxone}{\pq} \neq \frac{\oxtwo}{\pw}, \\
            X_A \leq \mysqrt{\frac{\oxone \oxtwo \pq}{\pw}}
        \end{gathered}},
    \end{equation}
    then according to \eqref{eq:adversary_payoff}, every unit of budget that the Adversary allocates towards Lotto game 1 has a greater marginal payoff than if it were allocated to Lotto game 2. Hence, the Adversary allocates their entire budget towards game 1 ($\xai{1} = X_A$).

\vspace{.2cm}
\noindent\emph{Case 2:} If $\OG$ belongs to
    \begin{equation}\label{eq:case_2}
        \C_{1 \leq 2}^2 \triangleq \lcurl{\OG \in \mathcal{G}_{1 \leq 2} \suchthat 0 < X_A - \mysqrt{\frac{\oxone \oxtwo \pq}{\pw}} \leq \oxtwo},
    \end{equation}
    then the Adversary's budget is sufficiently large such that it is optimal to allocate $\xai{1} \geq \oxone$, thereby hitting the range of diminishing returns (see \eqref{eq:adversary_payoff}) in Lotto game 1. Thus, they allocate budget towards Lotto game 2 such that the marginal payoffs in each game are equal, yielding $\xai{1} = \mysqrt{\frac{\oxone \oxtwo \pq}{\pw}}$.

\vspace{.2cm}
\noindent\emph{Case 3:} If $\OG$ belongs to
    \begin{equation}\label{eq:case_3}
        \C_{1 \leq 2}^3 \triangleq \lcurl{\OG \in \mathcal{G}_{1 \leq 2} \suchthat X_A - \mysqrt{\frac{\oxone \oxtwo \pq}{\pw}} > \oxtwo},
    \end{equation}
    then the Adversary's budget is larger than both Players' combined (in particular, $X_1 + X_2 < X_A$). Thus, the Adversary equates the marginal payoffs in the range of diminishing returns in both games, which yields $\xai{1} = \frac{X_A \mysqrt{\oxone \pq}}{\mysqrt{\oxone \pq} + \mysqrt{\oxtwo \pw}}$.

\vspace{.2cm}
\noindent\emph{Case 4:} If $\OG$ belongs to
    \begin{equation}\label{eq:case_4}
        \C_{1 \leq 2}^4 \triangleq \lcurl{\OG \in \mathcal{G}_{1 \leq 2} \suchthat \frac{\oxone}{\pq} = \frac{\oxtwo}{\pw} \text{ and } X_1 + X_2 \geq X_A},
    \end{equation}
    then the marginal payoff in the two games is already equal. Thus, the Adversary is indifferent between the two games, and any allocation satisfying $\xai{1} + \xai{2} = X_A$ and $\xai{i} \leq \overline{X}_i$, $i \in \{1, 2\}$, is optimal.

\subsection{Results for the Nominal Game}\label{app:nominal}

Throughout the remaining appendices, we consider games satisfying \eqref{eq:ratio-assumption}, and we write
\begin{equation*}
    Y \triangleq \mysqrt{\frac{\phi_1 X_1 X_2}{\phi_2}}
\end{equation*}
for the Adversary's allocation to game 1 in Case 2. The game lies in Case 1 if $X_A \le Y$, in Case 2 if $Y < X_A \le X_2 + Y$, and in Case 3 if $X_A > X_2 + Y$. Lemma~\ref{lem:nominal} covers Cases 2 and 3; we also use the following facts from \cite{kovenock2012coalitional}.

\begin{lemma}\label{lem:nominal-extra}
    \begin{enumerate}[label=(\roman*)]
        \item A Player does not gain from giving resources if, after the transfer, their budget-to-valuation ratio is at most their ally's. In particular, under \eqref{eq:ratio-assumption}, Player 1 never gains from giving resources.
        \item Suppose \eqref{eq:ratio-assumption} holds strictly. In Case 1, $G$ has no mutually beneficial transfer.
        \item On the boundary between Cases 2 and 3, i.e., when $X_A = X_2 + Y$, conditions \eqref{eq:nominal-case-2} and \eqref{eq:nominal-case-3} are equivalent.
    \end{enumerate}
\end{lemma}

\begin{proof}
    Parts (i) and (ii) are shown in the proofs of Propositions 1 and 2 of \cite{kovenock2012coalitional}. For (iii), substitute $X_A = X_2 + Y$ into \eqref{eq:nominal-case-2}. Since $Y (\phi_2 X_1/(\phi_1 X_2))^{1/2} = X_1$, the right-hand side times $X_A$ becomes $2X_1 + 2(\phi_2 X_1 X_2/\phi_1)^{1/2}$, so \eqref{eq:nominal-case-2} rearranges to $X_2 - X_1 > 2(\phi_2 X_1 X_2/\phi_1)^{1/2}$, which is \eqref{eq:nominal-case-3} multiplied by $X_2$.
\end{proof}

We refer to \eqref{eq:nominal-case-2} and \eqref{eq:nominal-case-3} as the \emph{Case 2 condition} and the \emph{Case 3 condition}. Both express that Player 2's payoff increases when they start giving resources. Specifically, writing $s = -\tau > 0$ for the amount Player 2 gives, so that $\oxone = X_1 + s$ and $\oxtwo = X_2 - s$, the derivatives of Player 2's payoff with respect to $s$ in Case 2 and Case 3 are \cite{kovenock2012coalitional}
\begin{align*}
    D_2 &= -\frac{\phi_2 X_A}{2 \oxtwo^2} + \frac{(\phi_1 \phi_2)^{1/2} (\oxone + \oxtwo)}{4\, \oxone^{1/2}\, \oxtwo^{3/2}}, \\
    D_3 &= \frac{1}{2X_A}\left( -\phi_2 + \frac{(\phi_1 \phi_2)^{1/2} (\oxtwo - \oxone)}{2 (\oxone \oxtwo)^{1/2}} \right),
\end{align*}
and rearranging shows that, at $s = 0$, $D_2 > 0$ is exactly the Case 2 condition and $D_3 > 0$ is exactly the Case 3 condition. Lemma~\ref{lem:nominal-extra}(iii) says that when a game crosses from Case 2 into Case 3, the condition governing the alliance changes its form but not its truth value. This is what allows us to handle uncertainty that spans both Cases.

\begin{lemma}\label{lem:recipient}
    In any game, if a transfer strictly increases the payoff of the Player who gives resources, then it also strictly increases the payoff of the Player who receives them.
\end{lemma}

\begin{proof}
    Let Player 2 be the one giving, so $\tau < 0$, and write $s = -\tau$. By Lemma~\ref{lem:nominal-extra}(i), if Player 2 gains, then after the transfer $\oxtwo/\phi_2 > \oxone/\phi_1$. Giving resources makes Player 2's ratio steadily smaller and Player 1's steadily larger, so the same strict inequality holds after every smaller transfer $s' \in [0, s]$. For each such $s'$, the Adversary's best response is in Case 1, 2, or 3, and Player 1's payoff is continuous in $s'$ and strictly increasing within each Case:
    \begin{itemize}[leftmargin=*, topsep=2pt, itemsep=1pt]
        \item In Case 1, the Adversary sends its whole budget to game 1, and Player 1's payoff \eqref{eq:player_payoff} increases with their own budget.
        \item In Case 2, Player 1's payoff is $\frac{1}{2}(\phi_1 \phi_2 \oxone / \oxtwo)^{1/2}$ \cite{kovenock2012coalitional}, which increases as $\oxone$ grows and $\oxtwo$ shrinks.
        \item In Case 3, the derivative of Player 1's payoff with respect to $s'$ is
        \begin{equation*}
            \frac{1}{2X_A}\left( \phi_1 + \frac{(\phi_1 \phi_2)^{1/2}(\oxtwo - \oxone)}{2(\oxone \oxtwo)^{1/2}} \right) \ge \frac{\phi_1}{4X_A} > 0,
        \end{equation*}
        where the inequality follows from $\oxone/\phi_1 < \oxtwo/\phi_2$.
    \end{itemize}
    Hence, Player 1's payoff strictly increases. The case in which Player 1 gives is identical with the indices swapped.
\end{proof}

\subsection{Proofs for Adversary Uncertainty}\label{app:adversary}

\begin{proof}[Proof of Theorem~\ref{thm:adversary-uncertainty}]
    Since $G$ has a mutually beneficial transfer, \eqref{eq:ratio-assumption} holds strictly by Lemma~\ref{lem:nominal-extra}(i), and this ordering does not depend on the Adversary's budget. Throughout, $X_A' \in [X_A(1 - \delta), X_A(1 + \delta)]$ denotes a possible Adversary budget, i.e., one of the games in ${\cal G}^{X_A}(G, \delta)$.

    \emph{Step 1 (one Adversary budget at a time).} A $\delta$-robust transfer is mutually beneficial in every game in ${\cal G}^{X_A}(G, \delta)$, so each of these games must have a mutually beneficial transfer. Conversely, suppose each of them has one. Then Player 2's marginal payoff at $\tau = 0$ is positive in each of these games (by Lemma~\ref{lem:nominal}, and by Lemma~\ref{lem:nominal-extra}(iii) on the Case 2/3 boundary). Since it varies continuously with $X_A'$ and $\tau$, it stays above some positive level over the closed interval of possible Adversary budgets for all small transfers, so every $\tau < 0$ with $|\tau|$ small enough benefits Player 2, and by Lemma~\ref{lem:recipient} also Player 1, in every game in ${\cal G}^{X_A}(G, \delta)$. Hence, it suffices to determine when every game in ${\cal G}^{X_A}(G, \delta)$ has a mutually beneficial transfer.

    \emph{Step 2 (how each requirement depends on $X_A'$).} By Lemmas~\ref{lem:nominal} and~\ref{lem:nominal-extra}(ii), the game with Adversary budget $X_A'$ must avoid Case 1, i.e., $X_A' > Y$, and must satisfy the condition of its Case.
    \begin{itemize}[leftmargin=*, topsep=2pt, itemsep=1pt]
        \item $Y$ does not depend on the Adversary's budget, so avoiding Case 1 is hardest for the weakest Adversary. It holds for every $X_A'$ if and only if $X_A(1 - \delta) > Y$, which is exactly $\delta < \alpha_1(G)$.
        \item The Case 2 condition \eqref{eq:nominal-case-2} becomes harder as the Adversary's budget grows, so it holds at budget $X_A'$ if and only if $X_A' < X_A^{\rm F}$, where $X_A^{\rm F} \triangleq X_A(1 + \alpha_2(G))$ is the budget at which it becomes tight. In particular, $X_A(1 + \delta) < X_A^{\rm F}$ is exactly $\delta < \alpha_2(G)$.
        \item The Case 3 condition \eqref{eq:nominal-case-3} does not depend on the Adversary's budget.
    \end{itemize}

    \emph{Step 3 (the boundary between Cases 2 and 3).} Budgets $X_A' \le X_2 + Y$ lie in Case 2, and budgets $X_A' > X_2 + Y$ lie in Case 3. By Lemma~\ref{lem:nominal-extra}(iii), the Case 2 condition holds on the boundary $X_A' = X_2 + Y$ if and only if the Case 3 condition holds; equivalently, $X_A^{\rm F} > X_2 + Y$ if and only if the Case 3 condition holds. Note that the Case 3 condition is exactly the condition $X_2 - X_1 > 2(\phi_2 X_1 X_2/\phi_1)^{1/2}$ in \eqref{eq:adversary-uncertainty-cond-1}.
    \begin{itemize}[leftmargin=*, topsep=2pt, itemsep=1pt]
        \item If the Case 3 condition holds (first line of \eqref{eq:adversary-uncertainty-cond-1}), then every budget in Case 2 satisfies $X_A' \le X_2 + Y < X_A^{\rm F}$, and every budget in Case 3 satisfies the Case 3 condition. The only remaining requirement is $\delta < \alpha_1(G)$.
        \item If the Case 3 condition fails (second line), then no possible budget may lie in Case 3, and every budget in Case 2 must satisfy $X_A' < X_A^{\rm F} \le X_2 + Y$. Both hold if and only if $X_A(1 + \delta) < X_A^{\rm F}$, i.e., $\delta < \alpha_2(G)$, which together with $\delta < \alpha_1(G)$ gives the second line. \qedhere
    \end{itemize}
\end{proof}

\subsection{Proofs for Ally Uncertainty}\label{app:ally}

\begin{lemma}\label{lem:donor}
    \begin{enumerate}[label=(\roman*)]
        \item No transfer $\tau > 0$ is $\delta$-robust to ally uncertainty.
        \item A transfer $\tau < 0$ benefits Player 2 in every game in ${\cal G}^{X_1}(G, \delta, \tau)$ only if $\delta < \beta_1(G)$.
        \item Given $\delta < \beta_1(G)$, a transfer $\tau < 0$ is $\delta$-robust to ally uncertainty if and only if it benefits Player 2 in every game in the set ${\cal G}^{X_1}(G, \delta)$.
    \end{enumerate}
\end{lemma}

\begin{proof}
    \emph{(i) Player 1 cannot be the donor.} If $\tau > 0$, Player 1 gives resources, and robustness requires Player 1 to gain in every game in their uncertainty set, which contains the true game $G$. In $G$, Player 1's ratio is at most Player 2's by \eqref{eq:ratio-assumption}, and giving away resources only lowers it further, so Player 1 cannot gain by Lemma~\ref{lem:nominal-extra}(i).

    \emph{(ii) Player 2 must be the harder target in every game they consider possible.} If $\tau < 0$, Player 2 gives resources and their uncertainty set is not refined, so robustness requires Player 2 to gain in every game in ${\cal G}^{X_1}(G, \delta)$, in particular in the one where Player 1's budget is largest, $X_1(1+\delta)$. If in that game $X_2/\phi_2 \le X_1(1+\delta)/\phi_1$, Player 2 cannot gain by Lemma~\ref{lem:nominal-extra}(i). Hence $X_2/\phi_2 > X_1(1+\delta)/\phi_1$ is necessary, and this is exactly $\delta < \beta_1(G)$.

    \emph{(iii) Player 1's condition holds automatically.} For $\tau < 0$, robustness asks Player 1 to gain in every game in ${\cal G}^{X_2}(G, \delta, \tau)$. In each of these games, Player 2's contests are strictly the harder target after the transfer, so by the argument in the proof of Lemma~\ref{lem:recipient}, Player 1's payoff strictly increases. Robustness therefore reduces to Player 2's condition: the transfer must benefit Player 2 in every game in ${\cal G}^{X_1}(G, \delta)$, which by Lemma~\ref{lem:recipient} is the same as being mutually beneficial in each of these games.
\end{proof}

Next, we show that if a mutually beneficial transfer exists when Player 1's budget is $X$, then one also exists for every smaller budget of Player 1. Only Player 1's budget varies, and we write $Y(x) \triangleq (\phi_1 x X_2/\phi_2)^{1/2}$ for the Case 1 threshold when Player 1's budget is $x$.

\begin{lemma}\label{lem:monotone}
    If the game $G|_{X_1 \mapsto X}$ has a mutually beneficial transfer, then so does $G|_{X_1 \mapsto X'}$ for every $X' < X$.
\end{lemma}

\begin{proof}
    Both $Y(x)$ and the Case 2/3 boundary $Y(x) + X_2$ increase with $x$. Lowering Player 1's budget from $X$ to $X'$ therefore never moves the game into Case 1 and can only move it from Case 2 into Case 3. We show that the condition of the Case in which $X'$ lies holds.

    \emph{If $X$ lies in Case 3,} then so does $X'$. Multiplying the Case 3 condition by $(X_2/x)^{1/2}$, it reads $(X_2 - x)/(x X_2)^{1/2} > 2(\phi_2/\phi_1)^{1/2}$, whose left-hand side decreases in $x$, so it holds at $X'$.

    \emph{If $X$ lies in Case 2,} consider Player 2's marginal payoff $D_2$ at $\tau = 0$ as a function of Player 1's budget $x$. Its first term does not depend on $x$, and its second depends on $x$ through $(x + X_2)/x^{1/2}$, which decreases in $x$ for $x < X_2$. This range applies: substituting $X_A > Y(X)$ into the Case 2 condition gives $X + X_2 > 2X$. Hence $D_2 > 0$, i.e., the Case 2 condition holds, at every budget between $X'$ and $X$. If $X'$ lies in Case 2, we are done. If $X'$ lies in Case 3, let $x_b$ be the budget between $X'$ and $X$ at which the game is on the Case 2/3 boundary. The Case 2 condition holds at $x_b$, so the Case 3 condition holds there by Lemma~\ref{lem:nominal-extra}(iii), and by the previous paragraph it also holds at $X' < x_b$.
\end{proof}

\begin{proof}[Proof of Theorem~\ref{thm:ally-uncertainty}]
    By Lemma~\ref{lem:donor}, $\delta < \beta_1(G)$ is necessary, and given it, a robust transfer exists if and only if a single $\tau < 0$ is mutually beneficial in every game in ${\cal G}^{X_1}(G, \delta)$, i.e., for every budget of Player 1 in $[X_1(1 - \delta), X_1(1 + \delta)]$. Let $X_{\max} \triangleq X_1(1 + \delta)$ be the largest budget of Player 1 that Player 2 considers possible.

    \emph{Step 1 (only the strongest Player 1 matters).} By Lemma~\ref{lem:monotone}, every game in ${\cal G}^{X_1}(G, \delta)$ has a mutually beneficial transfer if and only if $G|_{X_1 \mapsto X_{\max}}$ does, and by the continuity argument in the proof of Theorem~\ref{thm:adversary-uncertainty}, a single sufficiently small $\tau < 0$ then works in all of them. Hence, given $\delta < \beta_1(G)$, a robust transfer exists if and only if $G|_{X_1 \mapsto X_{\max}}$ has a mutually beneficial transfer.

    \emph{Step 2 (the requirements at $X_{\max}$).} By Lemmas~\ref{lem:nominal} and~\ref{lem:nominal-extra}(ii), this game has a mutually beneficial transfer if and only if the following hold.
    \begin{itemize}[leftmargin=*, topsep=2pt, itemsep=1pt]
        \item The game avoids Case 1: $X_A > Y(X_{\max})$. Squaring, this is exactly $\delta < \beta_2(G)$.
        \item If the game lies in Case 3, the Case 3 condition holds. Let $\kappa_0 \triangleq (\beta_5(G) + 1)/(\beta_5(G) - 1)$; as in the proof of Corollary~\ref{cor:ally-uncertainty}, $\kappa_0^{1/2} - \kappa_0^{-1/2} = 2(\phi_2/\phi_1)^{1/2}$. With $u = X_2/X_{\max}$, the Case 3 condition reads $u^{1/2} - u^{-1/2} > 2(\phi_2/\phi_1)^{1/2}$, and since the left-hand side increases in $u$, it holds if and only if $u > \kappa_0$, i.e., $\delta < \beta_6(G)$.
        \item If the game lies in Case 2, the Case 2 condition holds. With $w = (1 + \delta)^{1/2}$, so that $X_{\max} = X_1 w^2$, dividing the Case 2 condition by $X_1$ gives $w^2 - 2(\beta_2(G) + 1)^{1/2} w + X_2/X_1 > 0$, a quadratic in $w$. If $\beta_2(G) + 1 < X_2/X_1$, i.e., $X_A < X_2(\phi_1/\phi_2)^{1/2}$, it has no real roots and always holds. Otherwise, its smaller root is $\beta_3(G)$, and avoiding Case 1 reads $w < (\beta_2(G) + 1)^{1/2}$, the midpoint of the two roots. So, among games that avoid Case 1, the Case 2 condition holds if and only if $w < \beta_3(G)$, i.e., $\delta < \beta_4(G)$.
    \end{itemize}

    \emph{Step 3 (which Case applies).} Consider the budget $X_2/\kappa_0$ of Player 1, at which the Case 3 condition becomes tight. Since $\kappa_0 = (\beta_5(G) + 1)^2 \phi_2/\phi_1$ and $\phi_1/\phi_2 = \beta_5(G)^2 - 1$, the Case 2/3 boundary for this budget lies at the Adversary budget
    \begin{equation*}
        Y(X_2/\kappa_0) + X_2 = \frac{\phi_1 X_2}{\phi_2 (\beta_5(G) + 1)} + X_2 = \beta_5(G) X_2.
    \end{equation*}
    \begin{itemize}[leftmargin=*, topsep=2pt, itemsep=1pt]
        \item If $X_A > \beta_5(G) X_2$, the game with Player 1's budget $X_2/\kappa_0$ lies in Case 3 with its Case 3 condition tight, so it has no mutually beneficial transfer, and by Lemma~\ref{lem:monotone} neither does any game with a larger budget of Player 1. Every smaller budget places the game in Case 3 with the Case 3 condition satisfied, and also avoids Case 1. Hence, a robust transfer exists if and only if $\delta < \beta_6(G)$, which is the second line of \eqref{eq:ally-uncertainty-cond-1}; here $\beta_6(G) < \beta_2(G)$, and $\beta_6(G) < \beta_1(G)$ since $\kappa_0 > \phi_2/\phi_1$.
    \item Otherwise, $X_A \le \beta_5(G) X_2$. The Case 2/3 boundary $Y(x) + X_2$ increases with Player 1's budget $x$ and equals $\beta_5(G) X_2 \ge X_A$ at $x = X_2/\kappa_0$, so every budget of Player 1 in Case 3 is below $X_2/\kappa_0$ and, by Step 2, satisfies the Case 3 condition. On the boundary, the Case 3 condition holds (or is tight), so by Lemma~\ref{lem:nominal-extra}(iii) the Case 2 condition does too; since it holds exactly below a single threshold (Step 2), it also holds at every budget in Case 3. Hence, whether $X_{\max}$ lies in Case 2 or Case 3, the game $G|_{X_1 \mapsto X_{\max}}$ has a mutually beneficial transfer if and only if $\delta < \beta_2(G)$ and $\delta < \beta_4(G)$, which together with $\delta < \beta_1(G)$ give the first line of \eqref{eq:ally-uncertainty-cond-1}.
    \end{itemize}
\end{proof}

\bibliographystyle{plain}
\bibliography{ref}

\end{document}